\documentclass[conference]{IEEEtran}
\usepackage{amsfonts}
\usepackage{amssymb}
\usepackage{amsmath}
\usepackage{mathrsfs}
\usepackage{graphicx}
\usepackage{verbatim}
\usepackage{subfigure}
\usepackage{balance}
\usepackage{booktabs}
\usepackage{fancybox}
\usepackage{bm}
\usepackage{extarrows}
\usepackage{algorithm}
\usepackage{algorithmic}
\usepackage{multirow}
\usepackage{array}
\usepackage{epstopdf}
\usepackage{caption}
\usepackage{longtable}
\usepackage{rotating}
\usepackage{multirow}
\usepackage{diagbox}
\newtheorem{proof}{Proof}
\newtheorem{proposition}{Proposition}

\def \F {{}_2F_1}

\begin{document}
\title{Cooperative Transmission for Physical Layer Security by Exploring Social Awareness}
\author{
	
	\IEEEauthorblockN{Yiming Xu, Hui-Ming Wang, Qian Yang, Ke-Wen Huang, and Tong-Xing~Zheng}
	
	\IEEEauthorblockA{School of Electronic and Information Engineering, Xi'an Jiaotong University\\
		Ministry of Education Key Lab for Intelligent Networks and Network Security, Xi'an Jiaotong University\\
		Xi'an 710049, Shaanxi, P. R. China
	}

	\IEEEauthorblockA{
		Email: 469179938@qq.com
	}
}

\maketitle

\begin{abstract}
Social awareness and social ties are becoming increasingly fashionable with emerging mobile and handheld devices. Social trust degree describing the strength of the social ties has drawn lots of research interests in many fields including secure cooperative communications. Such trust degree reflects the users' willingness for cooperation, which impacts the selection of the cooperative users in the practical networks. In this paper, we propose a cooperative relay and jamming selection scheme to secure communication based on the social trust degree under a stochastic geometry framework. We aim to analyze the involved secrecy outage probability (SOP) of the system's performance. To achieve this target, we propose a double Gamma ratio (DGR) approach through Gamma approximation. Based on this, the SOP is tractably obtained in closed form. The simulation results verify our theoretical findings, and validate that the social trust degree has dramatic influences on the network's secrecy performance.
\end{abstract}

\IEEEpeerreviewmaketitle

\section{Introduction}
Nowadays, social ties have brought extensive influences among humankind. More and more people are actively involved in online social interactions \cite{SocialTie}, hence social ties among people are extensively broadened and significantly enhanced \cite{SocialTie4}. The so-called social ties are usually defined as the social relationships between individuals \cite{DefineSocialTie}, such as kinship, colleague relationships, friendship, acquaintance and so on \cite{SocialAware2}. The social trust degrees of the social ties among friends are the most basic and fundamental notions which characterize the strength of two individuals relating to each other \cite{SocialAware}. According to \cite{MeasureSocial}, ties have specific trust degree values describing the strength (i.e., from enmity to kinship) between the users. Moreover, social trust degree has drawn lots of research interests in various fields including mobile social networks, secure communications and so on.

Physical layer security (PLS) has drawn considerable attention during the past few years. Wyner's seminal research in \cite{Wyner} established a basic theory for the PLS. According to Wyner's theory, a positive secrecy capacity exists if the channel quality of the legitimate receiver is better than that of the eavesdropper. To protect the confidentiality of wireless transmissions, various communication technologies have been proposed, among which the user cooperation technology has been studied intensively.
As indicated by the survey paper \cite{BossTutarial}, various cooperative beamforming and jamming schemes have been proposed in \cite{JointInAF}--\cite{JointRobustAF}.
However, most of these existing works assume that the relays or the jammers have been chosen without considering the social trust degrees. Especially, whether each node should be chosen as relay or jammer according to the social trust degree has not been well studied.

The social ties of users reflect their willingness for sharing resources for safeguarding secrecy transmissions. In relay selection, the cooperative relay or jamming nodes should be selected according to their social trust degrees \cite{DefineSocialTie}. Recently, the social ties among users have been investigated in cooperative transmissions for PLS enhancement \cite{SecrecySocial1}-\cite{SecrecySocial4}.
Zheng \textit{et al.} \cite{SecrecySocial1} studied the secrecy rate and the secrecy throughput using average source-destination distance based on social ties.
To improve secrecy, Tang \textit{et al.} \cite{SecrecySocial2} discussed the secrecy outage probability (SOP) of a source-destination pair through the social tie based cooperative jamming game. A selection scheme based on mobility-impacted social interaction is proposed in \cite{SecrecySocial3} to maximize the worst-case secrecy rate in peer-to-peer (P2P) communications. An optimal cooperative transmission strategy is presented in \cite{SecrecySocial4} to maximize the secrecy rate, and the relays can be potential eavesdroppers according to their social trust degree. However, in these works, the nodes cooperate in either relay mode or jammer mode, which lacking the hybrid modes. The secrecy rate or the SOP have not been considered in terms of a stochastic geometry framework.

In this paper, we propose a cooperative relay and jamming scheme to secure wireless cooperation communications under a stochastic geometry framework. The nodes are classified into relays and jammers based on their locations and social trust degrees. We analyze the SOP to evaluate the system¡¯s security performance by applying a double Gamma ratio (DGR) approach. The DGR approach based on Gamma approximation facilitates mathematically tractable analysis and has a high accuracy.

Notations: $\mathbb{E}_A[\cdot]$ and $\mathbb{D}_A[\cdot]$ denote the mathematical expectation and variance with respect to a random variable $A$, respectively. $\mathcal{CN}(\mu, \sigma^2)$ denotes circularly symmetric complex Gaussian distribution with mean $\mu$ and variance $\sigma^2$. $\exp(1)$ denotes exponential distribution with mean 1. $\mathcal{A}(x,r)\subset \mathbb{R}^2$ denotes a bi-dimensional disk centered at $x$ with radius $r$. $\mathcal{D}(L_1,L_2)\subset \mathbb{R}^2$ denotes an annulus centered with internal radius $L_1$ and external radius $L_2$.

\section{Network Model and Problem Description}
As illustrated in Fig.~\ref{SCP_M_PHI}, we consider a wireless network over a finite circle area $\mathcal{A}(o,L_2)\subset \mathbb{R}^2$. This network consists of one source $s$, one destination $d$, one eavesdropper $z$, and a lot of legitimate nodes. Each node in the network works in a half-duplex mode and is equipped with a single antenna. The source hopes to transmit confidential signals to the destination without being wiretapped by the eavesdropper. Without loss of generality, we assume that the source $s$ is located at the origin $(0,0)$. The signal suffers from both small-scale fading and large-scale path loss. The channel between two nodes $x_1$ and $x_2$ is modeled as $H_{x_1,x_2} d_{x_1,x_2}^{-\alpha/2}$, where $H_{x_1,x_2}\thicksim\mathcal{CN}(0,1)$ is the quasi-static small-scale fading following the Rayleigh distribution, $d_{x_1,x_2}^{-\alpha/2}$ is the large-scale fading characterized by the standard path loss model $\|x_1 - x_2\|^{-\alpha/2}$ with $\alpha>2$ being the path-loss exponent \cite{Channel1}.

\subsection{Social Trust Degree Based Nodes Classification}
The locations of the legitimate nodes are modeled as a homogeneous poisson point process (PPP) $\Phi$ with density $\lambda$. We assume that the social trust degree of each legitimate node is independent and identically distributed (i.i.d.) and modeled as a uniform random variable $C$ distributed in $[0,1]$ as in \cite{MeasureSocial},\cite{SecrecySocial4}-\cite{RelaySocial3}. For a legitimate node, the trust degree is an increasing function of $C$, i.e., the source trusts the node more when $C$ increases.

We classify the legitimate nodes into relays, jammers and dummy nodes according to their locations as well as social trust degrees of the source.
The relays are the nodes whose trust degrees are in $[C_1,1]$ and are located within $\mathcal{A}(o,L_1)$. The trust degrees of the jammers are in $[C_2,C_1]$ and they located within $\mathcal{D}(L_1,L_2)$. Those nodes with trust degrees in $[0,C_2]$ are dummy nodes. According to the properties of the PPPs \cite{SGX}, the locations of the relays and the jammers are characterized by two independent and homogeneous PPPs denoted by $\Phi_R$ and $\Phi_J$ with densities $\lambda_R = (1 - C_1)\lambda$ and $\lambda_J = (C_1 - C_2)\lambda$, respectively. Throughout this paper, we use $x_R\in\Phi_R$ to denote the relays, and $x_J\in\Phi_J$ to denote the jammers.

\begin{figure}[!t]
\centering
\includegraphics[width=3.5in]{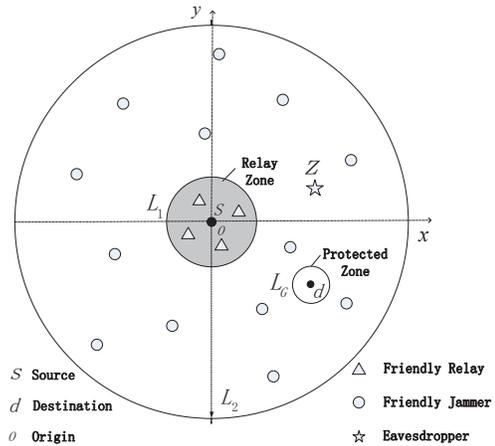}
\caption{Network model.}
\label{SCP_M_PHI}
\vspace{-5mm}
\end{figure}

\subsection{Cooperative Secrecy Transmission}
We assume all the relays work in the decode-and-forward (DF) mode with a two-phase transmission protocol. During the first phase of cooperative secrecy communication, source $s$ broadcasts confidential signal to the relays. We assume that the confidential information can be transmitted securely in this phase, which is  due to the following two reasons: 1) the source transmits with a sufficiently low power, such that the signal can not be decoded correctly by the eavesdropper outside $\mathcal{A}(o,L_1)$ due to the large-scale pass loss; 2) since $C_1$ is a sufficiently large threshold, the relays are the most trust nodes so that they will not leak the information to the eavesdropper. In this phase, we assume that the relays can always decode the received confidential information correctly, and the jammers keep silence. During the second phase, the relays forward the information to the destination. Since the destination is far away from the relays, the risk of being wiretapped is large. To safeguard the security, the jammers transmit jamming signals concurrently.
In order to protect the destination from being jammed, we set a protected zone $\mathcal{A}(d,L_G)$ \cite{ProtectedZone}.\footnote{~The destination node $d$ first broadcasts a pre-designed pilot signal with a pre-designed power. If a jammer receives the pilot signal, it is in the protect zone and it will not transmit jamming signals.} The jammers within this zone will keep silence at this phase. The dummy nodes will not take part in the confidential transmission in both the two phases.

\subsection{Signal Model}
We consider a cooperative beamforming scheme, where each relay transmits signal $s$ by pre-compensating the phase of the channel $H_{x_R,y}$. The transmitted symbol of each relay node is denoted as $s_{x_R} = \frac{\sqrt{P_s}H_{x_R,d}^*}{\|H_{x_R,d}\|}s$, where $P_s$ is the transmission power of the relay, $s$ is complex Gaussian distribution with mean 0 and the power of $s$ is unitary i.e., $E\{|s|^2\}=1$. Note that, such a cooperative beamforming scheme is a distributed one in the sense that each relay performs the cooperation with its own channel state information (CSI) instead of the global CSI. Consequently, the network overhead is greatly reduced.

Accordingly, the received signal at the destination $d$ is given by
\begin{align}
S_d(d) & = \sum_{x_R\in\Phi_R}\frac{\sqrt{P_s}H_{x_R,d}H_{x_R,d}^*}{\|H_{x_R,d}\|}d_{x_R,d}^{-\alpha/2}\cdot s \nonumber \\
& = \sum_{x_R\in\Phi_R}\sqrt{P_s}\|H_{x_R,d}\|d_{x_R,d}^{-\alpha/2}\cdot s . \label{Ty}
\end{align}
Similarly, the signal received by the eavesdropper is given by
\begin{align}
S(z) & = \sum_{x_R\in\Phi_R}\frac{\sqrt{P_s}H_{x_R,z}H_{x_R,d}^*}{\|H_{x_R,d}\|}d_{x_R,z}^{-\alpha/2}\cdot s .\label{Tz}
\end{align}

In order to enhance the security, the jammers also transmit independent Gaussian distributed interference signals to confuse the eavesdropper. Since the jamming signals from different jammers are independent, the aggregate jamming power received at the destination $d$ and the eavesdropper $z$ are given by
\begin{align}
I_d(d) = \sum_{x_J\in\mathcal{\overline{D}}}P_jh_{x_J,d}d_{x_J,d}^{-\alpha} \label{Iy}
\end{align}
and
\begin{align}
I(z) = \sum_{x_J\in\mathcal{\overline{D}}}P_jh_{x_J,z}d_{x_J,z}^{-\alpha}, \label{Iz}
\end{align}
respectively, where $P_j$ is the transmission power of each jammer, $\mathcal{\overline{D}}$ denotes the area $\Phi_J\backslash\mathcal{A}(d,L_G)$, and $h_{x_1,x_2}\thicksim \exp(1)$ is the power fading between locations $x_1$ and $x_2$. For analytical tractability, we focus on the interference-limited regime and ignore the noise at the receiver. The signal-to-interference ratio (SIR) at destination $d$ is given by
\begin{align}
SIR_d = \frac{\left|\sum_{x_R\in\Phi_R}\sqrt{P_s}\|H_{x_R,d}\|d_{x_R,d}^{-\alpha/2}\right|^2}{\sum_{x_J\in\mathcal{\overline{D}}}P_jh_{x_J,d}d_{x_J,d}^{-\alpha}}, \label{SIRy}
\end{align}
and the SIR for eavesdropper located at $z$ is given by
\begin{align}
SIR_z = \frac{\left|\sum_{x_R\in\Phi_R}\frac{\sqrt{P_s}H_{x_R,z}H_{x_R,d}^*}{\|H_{x_R,d}\|}d_{x_R,z}^{-\alpha/2}\right|^2}{\sum_{x_J\in\mathcal{\overline{D}}}P_jh_{x_J,z}d_{x_J,z}^{-\alpha}}. \label{SIRz}
\end{align}

\subsection{Performance Metric}
We use the SOP as a metric to evaluate the system secrecy performance. The SOP is defined as the probability that the SIR achieved by the single eavesdropper is larger than some threshold (i.e., the target SIR) $\beta_e$ \cite{SOPbyTXZheng}. Therefore, the SOP is given by
\begin{align}
\mathcal{P}_{so} & = \mathbb{P}\left\{SIR_z>\beta_e\right\} = 1 - \mathbb{P}\left\{\frac{T(z)}{I(z)}\leq\beta_e\right\},
\label{Pso2}
\end{align}
where $T(z) = \left|\sum_{x_R\in\Phi_R}\frac{\sqrt{P_s}H_{x_R,z}H_{x_R,d}^*}{\|H_{x_R,d}\|}d_{x_R,z}^{-\alpha/2}\right|^2$, $\frac{H_{x_R,z}H_{x_R,d}^*}{\|H_{x_R,d}\|}d_{x_R,z}^{-\alpha/2}\thicksim \mathcal{CN}(0,d_{x_R,z}^{-\alpha})$ which is independent for arbitrary $x_R\in\Phi_R$. Although $T(z)$ is conditional exponential distributed with conditional mean $P_s\sum_{x_R\in\Phi_R}d_{x_R,z}^{-\alpha}$, the mean is related to the locations of $x_R$ in $\Phi_R$. This makes it untractable to calculate the probability in \eqref{Pso2} through common ways. Also it is complicated to derive the probability distribution function (PDF) of $I(z)$.

In order to obtain a closed form expression of the SOP, we propose a DGR approach to make \eqref{Pso2} mathematically tractable, which will be discussed in Section III.

\section{DGR Approach for Calculating SOP} \label{SectionDGR}
In this paper, we aim to analyze the SOP in our considered network according to \eqref{Pso2}. However, in our scheme, it is complicated to derive the PDFs of $T(z)$ and $I(z)$.
To tackle this problem, we propose the following DGR approach, which will facilitate our calculations.

In order to describe the DGR approach, we primarily introduce the Gamma approximation. The Gamma approximation is an approach to approximate the distribution of a random variable based on the Gamma distribution \cite{GamaAppr}. By employing the Gamma approximation, the PDFs of $T(z)$ and $I(z)$ can be easily obtained. The Gamma approximated PDF of a random variable $A$ is given as
\begin{align}
G_A(x_A;\nu_A,\theta_A) & = \frac{x_A^{\nu_A - 1}e^{-\frac{x_A}{\theta_A}}}{\theta_A^{\nu_A}\Gamma(\nu_A)}, \label{Gama}
\end{align}
where $\Gamma(\nu_A)$ is the Gamma function \cite[eq(6.45)]{TableIntegral}, and the parameters $\nu_A$ and $\theta_A$ are derived from matching the first and second order moments of $A$. The $i$-th cumulants $N_A^{(i)}$ of variable $A$ is defined as
\begin{align}
N_A^{(i)} & = \frac{\mathrm{d}^{i}\mathbb{E}_A\left[e^{wa}\right]}{\mathrm{d}w^{i}}\Big|_{w = 0}. \label{Eti}
\end{align}
The mean of $A$ is thereby denoted as $\mu_A = N_A^{(1)}$, and the variance of $A$ is given by $\sigma_A^2 = N_A^{(2)} - \left(N_A^{(1)}\right)^2$. Accordingly, the approximated Gamma variable has the parameters
\begin{align}
\nu_A = \frac{\mu_A^2}{\sigma_A^2},\quad\theta_A = \frac{\sigma_A^2}{\mu_A}. \label{vt}
\end{align}

According to \eqref{Gama}-\eqref{vt}, we can obtain the approximated PDFs $G_T(x_T;\nu_T,\theta_T)$ and $G_I(x_I;\nu_I,\theta_I)$ for $T(z)$ and $I(z)$ in our network model, respectively.
Firstly, we derive the parameters $\nu_{T}$, $\theta_{I}$, $\nu_{I}$, and $\theta_{I}$ as
\begin{align}
& \nu_{T} = \frac{\lambda_RQ_z(1)}{\lambda_RQ^2_z(1)+2Q_z(2)},\quad\theta_{T} = \frac{P_s[\lambda_RQ^2_z(1)+2Q_z(2)]}{Q_z(1)}, \label{nuT}
\end{align}
\begin{align}
& \nu_{I} = \frac{\lambda_J\big(\int_{\mathcal{\overline{D}}}\frac{1}{d_{x_J,z}^{\alpha}}\mathrm{d}x_J\big)^2}{2\int_{\mathcal{\overline{D}}}\frac{1}{d_{x_J,z}^{2\alpha}}\mathrm{d}x_J},\quad
\theta_{I} = \frac{2P_j\int_{\mathcal{\overline{D}}}\frac{1}{d_{x_J,z}^{2\alpha}}\mathrm{d}x_J}{\int_{\mathcal{\overline{D}}}\frac{1}{d_{x_J,z}^{\alpha}}\mathrm{d}x_J}, \label{nuI}
\end{align}
where $Q_z(n) = \int_{\mathcal{A}(o,L_1)}d_{x_R,z}^{-n\alpha}\mathrm{d}x_R$. Due to the conditional exponential distribution of $T(z)$ and the space limitation, we only provide the derivation details of \eqref{nuI} in Appendix A.
Then the approximated PDFs of $T(z)$ and $I(z)$ are given by
\begin{align}
G_{T}(x_T;\nu_{T},\theta_{T}) & = \frac{x_T^{\nu_{T} - 1}e^{-x_T/\theta_{T}}}{\theta_{T}^{\nu_{T}}\Gamma(\nu_{T})},
\end{align}
and
\begin{align}
G_{I}(x_I;\nu_{I},\theta_{I}) & = \frac{x_I^{\nu_{I} - 1}e^{-x_I/\theta_{I}}}{\theta_{I}^{\nu_{I}}\Gamma(\nu_{I})},
\end{align}
respectively.

After obtaining the approximated PDFs $G_T(x_T;\nu_T,\theta_T)$ and $G_I(x_I;\nu_I,\theta_I)$, our objective SOP in \eqref{Pso2} is giving in the following proposition.

\begin{proposition}[The DGR Approach]
The SOP derived from the ratio of two approximated Gamma variables is given by
\begin{align}
\mathcal{P}_{so} & = \frac{q_e^{\nu_{T}}\Gamma(\nu_{T} + \nu_{I})}{\nu_{I}(q_e + 1)^{\nu_{T} + \nu_{I}}\Gamma(\nu_{T})\Gamma(\nu_{I})}\cdot F, \label{Pso3}
\end{align}
where $q_e = \frac{\beta_e\theta_{I}}{\theta_{T}}$, $F = \F\left(1,\nu_{T} + \nu_{I};\nu_{I} + 1;\frac{1}{q_e + 1}\right)$ with $\F$ being hypergeometric function \cite[Eq. 6.455.1]{TableIntegral}.
\end{proposition}

\begin{proof}
According to the formulation of $\mathcal{P}_{so}$ in \eqref{Pso2}
\begin{align}
&\quad~~\mathbb{P}\left\{\frac{T(z)}{I(z)}\leq\beta_e\right\} \nonumber \\
& = \mathbb{P}\left\{T(z)\leq\beta_e I(z)\right\} \nonumber \\
& \overset{(a)}{=} \mathbb{E}_{I(z)}\left[\int_0^{\beta_e I}G_T(x_T;\nu_T,\theta_T)\mathrm{d}x_T\right] \nonumber \\
& \overset{(b)}{=} \int_0^{\infty}\left(1 - \frac{\Gamma\left(\nu_T,\frac{\beta_e I}{\theta_T}\right)}{\Gamma(\nu_T)}\right)G_I(x_I;\nu_I,\theta_I)\mathrm{d}x_I \nonumber \\
& \overset{(c)}{=} 1 - \frac{q_e^{\nu_T}\Gamma(\nu_T + \nu_I)}{\nu_I(q_e + 1)^{\nu_T + \nu_I}\Gamma(\nu_T)\Gamma(\nu_I)}\cdot F, \nonumber
\end{align}
where $(a)$ follows from the total probability formulation, $(b)$ follows from the definition of the incomplete Gamma function \cite[Eq. 6.45]{TableIntegral} and the approximated PDF of $I(z)$. After some integral calculations, $(c)$ follows from applying \cite[Eq. 6.455.1]{TableIntegral}.
Consequently,
\begin{align}
\mathcal{P}_{so} & = 1 - \mathbb{P}\left\{\frac{T(z)}{I(z)}\leq\beta_e\right\} \nonumber \\
& = \frac{q_e^{\nu_{T}}\Gamma(\nu_{T} + \nu_{I})}{\nu_{I}(q_e + 1)^{\nu_{T} + \nu_{I}}\Gamma(\nu_{T})\Gamma(\nu_{I})}\cdot F. \nonumber
\end{align}
\end{proof}

\begin{figure}[!t]
\centering
\includegraphics[width=3in]{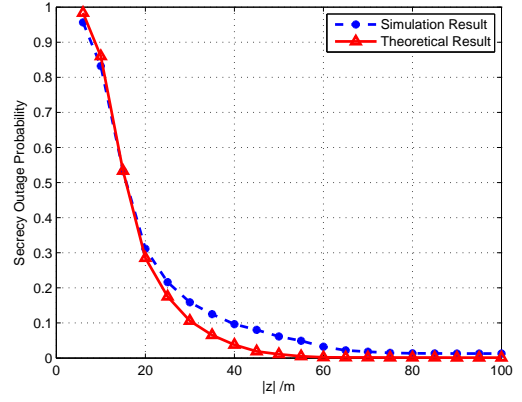}
\caption{SOP vs. different distances $d_E$ of the eavesdropper. The system parameters are $\beta_e=0$ dB, $\alpha=4$, $L_1=6$ m, $L_2=100$ m, $L_G=5$ m, $C_1=0.8$, $C_2=0.79$, $\lambda=0.2 ~/\mathrm{m}^2$, $P_s=10$ dBm, and $P_j=1$ dBm.}
\label{SOP1}
\end{figure}
Fig.~\ref{SOP1} depicts the theoretical result in Proposition 1 and the simulation result of the SOP, where $d_E$ denotes the distance between the origin and the eavesdropper. We can see from Fig.~\ref{SOP1} that the proposed DGR approach has high accuracy.


\section{Numerical Results and Discussions}
In this section, representative numerical results are presented to illustrate the SOP performance in our network model. Considering the accuracy of the DGR approach and also for simplicity, we only present the numerical results based on Proposition 1. As the legitimate nodes in the network are classified into relays and jammers according to their social trust degrees of the source, the density $\lambda_R$ of the relays and the density $\lambda_J$ of the jammers are determined by $(1 - C_1)\lambda$ and $C_q\lambda$, respectively, where $C_q$ denotes $(C_1 - C_2)$. We mainly focus on the impacts of $C_1$ and $C_q$ on the SOP. The other system parameters as set as follows: the fading power exponent is $\alpha=4$, the border of $\mathcal{A}(0,L_1)$ is $L_1=6$ m, the outer border of $\mathcal{D}(L_1,L_2)$ is $L_2=100$ m, the radius of the protected zone is $L_G=5$ m, the density of the legitimate nodes is $\lambda=0.2 ~/\mathrm{m}^2$, the transmission power of the relay is $P_s=10$ dBm, and the transmission power of the jammer is $P_j=1$ dBm.

\begin{figure}[!tp]
\centering
\includegraphics[width=3in]{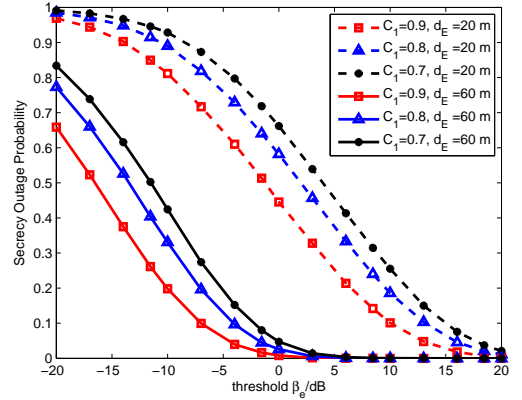}
\caption{SOP of single eavesdropper vs. $\beta_e$ for various social trust degree $C_1$ and distances $d_E$ of the eavesdropper.}
\label{sop_dfnZ_lamR_rE}
\end{figure}

Fig.~\ref{sop_dfnZ_lamR_rE} plots the SOP versus $\beta_e$ for various social trust degrees $C_1$ of the source and distances $d_E$ of the eavesdropper.
Comparing the curves with the same $d_E$, we see that as $C_1$ increases, the SOP decreases. This is because a larger $C_1$ is equivalent to a smaller $\lambda_R$, which results in less relays producing lower $SIR_E$. Therefore, there is a higher probability for performing perfect secrecy, which leads to a lower SOP. We also see that the SOP decreases with increasing $d_E$ by the comparison among the curves with the same $C_1$. This is due to the fact that the secrecy outage occurs more frequently when the distance between the source and the eavesdropper decreases. Since $C_1$ represents the trust degree of the source, we know that the most private message should be transmitted to the person with sufficiently high trust degree in order to realize perfect secrecy.

\begin{figure}[!tp]
\centering
\includegraphics[width=3in]{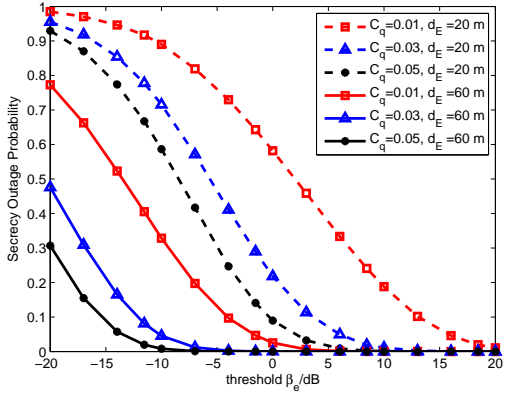}
\caption{SOP of single eavesdropper vs. $\beta_e$ for various $C_q$ and $d_E$.}
\label{sop_dfnZ_lmJ_rE}
\end{figure}

Fig.~\ref{sop_dfnZ_lmJ_rE} compares the SOP versus $\beta_e$ for various $C_q$ and $d_E$, where $C_1$ is set as 0.8. By comparing the curves with the same $d_E$, we see that as $C_q$ increases, the SOP decreases. This is because a larger $C_q$ is equivalent to a larger $\lambda_J$, which results in more jammers producing lower $SIR_E$. We also observe that the SOP dramatically decreases with increasing $d_E$. As $C_q$ derived from $(C_1-C_2)$ determines the density of the jammers, smaller $C_2$ will lead to more jammers offering intentional interference to improve the secrecy performance. From this we know that a diminishing social trust degree will disrupt the eavesdropper more efficiently.

\section{Conclusion}
In this paper, we proposed a cooperative relay and jamming scheme based on social trust degree to secure communications. In order to facilitate the analysis of the SOP, we proposed a DGR approach under a stochastic geometry framework. The simulation results have high accuracy and show that the social ties have dramatic influences on the system's secrecy performance. The proposed scheme has considered a practical scene in which the social trust degree is employed to reflect the users' willingness of cooperation. For further researches, we will study the COP and the secrecy throughput in the scheme, and consider the social trust degree further practically.

\appendices
\section{Derivation Details of \eqref{nuI}}
According to \eqref{Eti} and the definition of $I(z)$ in \eqref{Iz}, the $i$-th cumulants of $I(z)$ is given by
\begin{align}
N_I^{(i)} & = \frac{\mathrm{d}^i\mathbb{E}_{\Phi_J,h_{x_J,z}}\left[e^{w\sum_{x_J\in\mathcal{\overline{D}}}P_jh_{x_J,z}d_{x_J,z}^{-\alpha}}\right]}{\mathrm{d}w^i}\Big|_{w = 0} .\label{EIy}
\end{align}
First, we calculate $\mathbb{E}_{\Phi_J,h_{x_J,z}}\left[e^{w\sum_{x_J\in\mathcal{\overline{D}}}P_jh_{x_J,z}d_{x_J,z}^{-\alpha}}\right]$ as
\begin{align}
&\quad \mathbb{E}_{\Phi_J,h_{x_J,z}}\left[e^{w\sum_{x_J\in\mathcal{\overline{D}}}P_jh_{x_J,z}d_{x_J,z}^{-\alpha}}\right] \nonumber \\
& = \mathbb{E}_{\Phi_J,h_{x_J,z}}\bigg[\prod_{x_J\in\mathcal{\overline{D}}}e^{wP_jh_{x_J,z}d_{x_J,z}^{-\alpha}}\bigg] \nonumber \\
& \overset{(d)}{=} \mathbb{E}_{\Phi_J}\bigg[\prod_{x_J\in\mathcal{\overline{D}}}\mathbb{E}_{h_{x_J,z}}\left[e^{wP_jh_{x_J,z}d_{x_J,z}^{-\alpha}}\right]\bigg] \nonumber \\
& \overset{(e)}{=} \mathbb{E}_{\Phi_J}\bigg[\prod_{x_J\in\mathcal{\overline{D}}}\frac{1}{1 - wP_jd_{x_J,z}^{-\alpha}}\bigg] \nonumber \\
& \overset{(f)}{=} \exp\bigg[-\lambda_J\int_{\mathcal{\overline{D}}}\left(1 - \frac{1}{1 - wP_jd_{x_J,z}^{-\alpha}}\right)\mathrm{d}x_J\bigg] \nonumber \\
& = \exp\left(\lambda_J\int_{\mathcal{\overline{D}}}\frac{wP_j}{d_{x_J,z}^{\alpha} - wP_j}\mathrm{d}x_J\right), \label{AppB1}
\end{align}
where $(d)$ follows since $h_{x_J,z}$ is independent of $\Phi_J$. $(e)$ follows from $h_{x_J,z} \thicksim \exp(1)$, and by applying the probability generating functional (PGFL) of the PPP we can obtain $(f)$. Consequently, substituting \eqref{AppB1} into \eqref{EIy}, $N_I^{(1)}$ is given by
\begin{align}
&\quad N_I^{(1)} = \frac{\mathrm{d}\Big(\exp\big(\lambda_J\int_{\mathcal{\overline{D}}}\frac{wP_j}{d_{x_J,z}^{\alpha} - wP_j}\mathrm{d}x_J\big)\Big)}{\mathrm{d}w}\Big|_{w = 0} \nonumber \\
& = \exp\left(\lambda_JG_1\right)\cdot\lambda_JP_jG_2(2)\big|_{w = 0} \nonumber \\
& = \lambda_J\int_{\mathcal{\overline{D}}}\frac{P_j}{d_{x_J,z}^{\alpha}}\mathrm{d}x_J, \label{AppB2}
\end{align}
where $G_1 = \int_{\mathcal{\overline{D}}}\frac{wP_j}{d_{x_J,z}^{\alpha} - wP_j}\mathrm{d}x_J$, $G_2(n) = \int_{\mathcal{\overline{D}}}\frac{d_{x_J,z}^{\alpha}}{\left(d_{x_J,z}^{\alpha} - wP_j\right)^n}\mathrm{d}x_J$. Let $i = 2$, $N_I^{(2)}$ is given as
\begin{align}
&\quad N_I^{(2)} = \frac{\mathrm{d}^2\Big(\exp\big(\lambda_J\int_{\mathcal{\overline{D}}}\frac{wP_j}{d_{x_J,z}^{\alpha} - wP_j}\mathrm{d}x_J\big)\Big)}{\mathrm{d}w^2}\Big|_{w = 0} \nonumber \\
& \overset{(g)}{=} \lambda_JP_j^2\exp\left(\lambda_JG_1\right)\cdot\left[\lambda_JG^2_2(2) + 2G_2(3)\right]\big|_{w = 0} \nonumber \\
& = \lambda_J^2\bigg(\int_{\mathcal{\overline{D}}}\frac{P_j}{d_{x_J,z}^{\alpha}}\mathrm{d}x_J\bigg)^2 + 2\lambda_J\int_{\mathcal{\overline{D}}}\frac{P_j^2}{d_{x_J,z}^{2\alpha}}\mathrm{d}x_J, \label{AppB3}
\end{align}
where $(g)$ is the second-order differential results. As a result, by substituting \eqref{AppB2} and \eqref{AppB3}, $\sigma_I^2$ is given by
\begin{align}
& \quad \sigma_I^2 = N_I^{(2)} - \left(N_I^{(1)}\right)^2
= 2\lambda_J\int_{\mathcal{\overline{D}}}\frac{P_j^2}{d_{x_J,z}^{2\alpha}}\mathrm{d}x_J. \label{AppB4}
\end{align}

According to \eqref{vt}, $\nu_{I}$ and $\theta_{I}$ are obtained as
\begin{align}
\nu_{I} = \frac{\lambda_J\big(\int_{\mathcal{\overline{D}}}\frac{1}{d_{x_J,z}^{\alpha}}\mathrm{d}x_J\big)^2}{2\int_{\mathcal{\overline{D}}}\frac{1}{d_{x_J,z}^{2\alpha}}\mathrm{d}x_J}
\end{align}
and
\begin{align}
\theta_{I} = \frac{2P_j\int_{\mathcal{\overline{D}}}\frac{1}{d_{x_J,z}^{2\alpha}}\mathrm{d}x_J}{\int_{\mathcal{\overline{D}}}\frac{1}{d_{x_J,z}^{\alpha}}\mathrm{d}x_J},
\end{align}
respectively.

\section*{Acknowledgment}
This work was partially supported by the National Natural Science Foundation of China under Grants 61671364 and 61701390, the Author of National Excellent Doctoral Dissertation of China under Grant 201340, the Young Talent Support Fund of Science and Technology of Shaanxi Province under Grant 2015KJXX-01, and the China Postdoctoral Science Foundation under Grant 2017M613140.

\end{document}